\documentclass[11pt]{article}
\usepackage[latin9]{inputenc}
\usepackage{geometry}
\usepackage{amsmath}
\usepackage{amsthm}
\usepackage{amssymb}
\usepackage{setspace}
\usepackage[authoryear]{natbib}
\makeatletter

\DeclareTextSymbolDefault{\textquotedbl}{T1}

\theoremstyle{plain}
\newtheorem{lem}{\protect\lemmaname}
\theoremstyle{plain}
\newtheorem{cor}{\protect\corollaryname}
\theoremstyle{plain}
\newtheorem{prop}{\protect\propositionname}
\theoremstyle{definition}
\newtheorem{defn}{\protect\definitionname}

\AtBeginDocument{

}

\providecommand{\corollaryname}{Corollary}
\providecommand{\definitionname}{Definition}
\providecommand{\lemmaname}{Lemma}
\providecommand{\propositionname}{Proposition}

\usepackage{tikz}

\providecommand{\corollaryname}{Corollary}
\providecommand{\definitionname}{Definition}
\providecommand{\lemmaname}{Lemma}
\providecommand{\propositionname}{Proposition}

\providecommand{\corollaryname}{Corollary}
\providecommand{\definitionname}{Definition}
\providecommand{\lemmaname}{Lemma}
\providecommand{\propositionname}{Proposition}

\providecommand{\corollaryname}{Corollary}
\providecommand{\definitionname}{Definition}
\providecommand{\lemmaname}{Lemma}
\providecommand{\propositionname}{Proposition}

\makeatother

\providecommand{\corollaryname}{Corollary}
\providecommand{\definitionname}{Definition}
\providecommand{\lemmaname}{Lemma}
\providecommand{\propositionname}{Proposition}

\begin{document}

\title{\textbf{Tax Mechanisms and Gradient Flows }}

\author{\textbf{Stefan Steinerberger}\thanks{S.S. is supported by the NSF (DMS-1763179) and the Alfred P. Sloan
Foundation.}\textbf{ }\\
 {\small{}{}{}{}{}{}{}Yale University} \and \textbf{Aleh Tsyvinski}\\
 {\small{}{}{}{}{}{}{} Yale University}\thanks{We thank Manuel Amador, Andy Atkeson, Felix Bierbrauer, Patrick Bolton,
Jaroslav Borovicka, Pierre Boyer, Hector Chade, Alfred Galichon, Martin
Hellwig, Oleg Itskhoki, Ilse Lindenlaub, Benny Moldovanu, Chris Phelan,
Georgii Riabov, Larry Samuelson, Florian Scheuer, Philipp Strack,
Kjetil Storesletten, Juuso Valimaki, and, especially, Nicolas Werquin.}}
\maketitle
\begin{abstract}
We demonstrate how a static optimal income taxation problem can be
analyzed using dynamical methods. Specifically, we show that the taxation
problem is intimately connected to the heat equation. Our first result
is a new property of the optimal tax which we call the fairness principle.
The optimal tax at any income is invariant under a family of properly
adjusted Gaussian averages (the heat kernel) of the optimal taxes
at other incomes. That is, the optimal tax at a given income is equal
to the weighted by the heat kernels average of optimal taxes at other
incomes and income densities. Moreover, this averaging happens at
every scale tightly linked to each other providing a unified weighting
scheme at all income ranges. The fairness principle arises not due
to equality considerations but rather it represents an efficient way
to smooth the burden of taxes and generated revenues across incomes.
Just as nature wants to distribute heat evenly, the optimal way for
a government to raise revenues is to distribute the tax burden and
raised revenues evenly among individuals. We then construct a gradient
flow of taxes -- a dynamic process changing the existing tax system
in the direction of the increase in tax revenues -- and show that
it takes the form of a heat equation. The fairness principle holds
also for the short-term asymptotics of the gradient flow, where the
averaging is done over the current taxes. The gradient flow we consider
can be viewed as a continuous process of a reform of the nonlinear
income tax schedule and thus unifies the variational approach to taxation
and optimal taxation. We present several other characteristics of
the gradient flow focusing on its smoothing properties. 
\end{abstract}
\newpage{}

\section*{Introduction}

The classic static optimal nonlinear income taxation problem (Mirrlees
1971) is one of the prototypical examples of mechanism design problems.\footnote{See, e.g., the book by Salanie (2011).}
Diamond (1998) and Saez (2001) derive expressions for the optimal
taxes that are widely used in public finance. Yet, even in the simplest
cases, those do not provide closed-form solutions or a full characterizaton
for the optimal income tax. The solution is a coupled system of differential
equations and has relatively few known characteristics.\footnote{See, e.g., Hellwig (2010) for the state-of-the art analysis of incentive
problems with unidimensional characteristics. } 

A new aspect of our analysis is in the application of dynamical methods
to static mechanism design problems. Starting from a given tax (either
optimal or suboptimal), we associate with it a dynamical system --
the analysis of which yields a new view of the problem. This dynamical
system is tightly connected to a heat equation, one of the most well-studied
and well-behaved partial differential equations.

Our main result shows that the optimal tax satisfies a certain invariance
relationship which we call a fairness principle. The fairness principle
states that an optimal tax at a given income is equal to the properly
weighted average of the taxes at all other incomes. We call it ``fair''
not because the planner has a social welfare function that places
a weight on equality but rather that it is efficient to treat agents
at various incomes similarly. In this sense, fairness is a different
version of smoothing the burden of the deadweight losses of taxes
and raising revenues.

We start by analyzing the static optimal nonlinear income tax. We
associate with the equation determining the optimal tax a dynamical
object -- the heat kernel which determines a family of averaging
functions at different times (or scales). We show that the optimal
tax is invariant under the heat kernel. Specifically, the optimal
tax at any given income is determined as a weighted average of the
optimal taxes at other incomes and of income densities. Importantly,
we show that this averaging is done at \textit{all} times (scales).
What we mean by this is that as the heat kernel is spreading wider
(with, say, time), the optimal tax is equal to the average tax at
the wider sets of incomes and income densities: from the weight placed
on incomes just around the point to more and more incomes. In other
words, the planner at the optimum ensures that the deadweight loss
or the behavioral effect of taxes is distributed fairly -- as the
weighted average of taxes at all scales. This behavioral effect is
added to the statutory impact of taxes or the mechanical effect which
is also a weighted average over all scales of income densities. A
fact of crucial importance is that the fairness principles for different
scales are tightly linked via the so called semigroup property and
is thus far from arbitrary. That is, there is one unified weighting
scheme across all of the income scales. We then show that the behavior
of the heat kernel is essentially Gaussian. For any time (scale),
we provide a Gaussian upper bound that shows that the heat kernel
decays as the Gaussian. For the small time asymptotics, the heat kernel
is exactly a Gaussian scaled by the density and the elasticity at
the given income at the optimum. A corollary of the fairness principle
is a formula determining the optimal marginal tax by the global behavior
of the level of the optimal tax code. These results are new to the
taxation and mechanism design literature. 

Second, we construct a gradient flow for any initial (optimal or suboptimal)
tax function. The gradient flow is a dynamical system that changes
the underlying tax function in the direction of the steepest descent
of the tax revenue functional. We show that the fairness principle
also holds along the trajectory of the gradient flow. Specifically,
we prove that the small time behavior of the gradient flow (that is,
the short-time asymptotics) is such that the evolved tax is equal
to the Gaussian average of the current taxes at all incomes and income
densities, where the Gaussian is adjusted by the density and the elasticity
at the current tax system. The optimal tax is a stationary point of
the gradient flow. That is, starting from the optimum both the averaging
and the agents' behavior stays constant.

We then propose to use the operator-splitting methods corresponding
to the changes in the tax revenues and to the agent's behavior and
show that the gradient flow is a heat equation within each step of
operator splitting. The operator splitting technique has a natural
economic meaning. On the trajectory of the tax reform, the government,
for an infinitesimally short time, evolves taxes in the direction
of increased revenues, keeping the density of agents\textquoteright{}
incomes and elasticities fixed at their value observed in the current
economy. Then the densities and the elasticities are recomputed. That
is, the government evaluates the changes in revenues under the current
information given by the exogenous sufficient statistics evaluated
at a given time.\footnote{See, e.g., Kleven (2018) for a recent review of the sufficient statistics
approach in taxation.} Within each step of this process, the gradient flow of tax revenues
is a heat equation. Having identified the heat equation as the principal
object governing the evolution of the tax function, there is a wealth
of additional results one can obtain to study its behavior. We primarily
focus on showing that the gradient flow possesses a very nice underlying
smoothing structure. First, we show that the gradient flow is trying
to smooth out the rough irregularities in the tax schedule. That is,
it acts first in the income regions which have the largest deviation
of the marginal tax from the properly defined reference point. These
deviations are weighted by the natural weighing measure that depends
on the labor income elasticity and the density of the incomes. Second,
we show that the reform viewed as the gradient flow leads to the continuous
tax systems, mollifying any tax scheme instantaneously. Third, we
use the Sturm-Liouville theory to describe further characteristics
of the gradient flow and, in particular, to show exponential decay
of the variability of the tax system compared to the reference point.

The fact that the gradient flow gives rise to a law of motion for
the tax function characterized by a parabolic PDE (namely, a heat
equation) is subtle and is at the essence of the economics of the
problem. The well-known and much studied heat equation arises in a
new, and possibly surprising, setting: in an environment where a government
is interested in collecting income tax revenue and the agents are
heterogeneous in their skills. In a broad sense, the agents and the
government have opposing objectives: to maximize revenue for the government,
and to maximize the utility for the agents, given taxes. We show,
among other things, that a local process where the government tries
to maximize its income via taxation leads, in a natural way, to the
heat equation. Just as nature wants to distribute heat evenly on a
background (leading to the spreading of heat), the optimal way for
a government to raise revenues is to distribute the tax burden and
raised revenues evenly among individual subjects -- this is perhaps
particularly surprising. This fairness is indeed a phenomenon coming
from agents adapting to uneven taxation systems by working more or
less, respectively. We believe this to be a new way of arriving at
the heat equation and a new tax and revenue smoothing result in taxation.

The construction of the dynamical system that we propose also contributes
to unifying the optimal and the variational approach to taxation.
The variational approach considers a potentially suboptimal tax and
proceeds with varying it locally to derive the formulas for the effects
of the tax reforms.\footnote{See, e.g., Saez (2001), Kleven and Kreiner (2006), and Golosov, Tsyvinski,
and Werquin (2014) for the methodology; Kleven, Kreiner, and Saez
(2009) and Jacquet and Lehmann (2015) for the analysis of the multidimensional
types; Saez and Stantcheva (2016) and Bierbrauer and Boyer (2018)
for the political economy context; Sachs, Tsyvinski, and Werquin (2016)
and Scheuer and Werning (2016) for the analysis in general equilibrium;
and Saez and Stantcheva (2018) for capital income taxation.} When considering the effects of the local tax reform, a natural question
to ask is whether there is a process of reforms that may lead to the
optimum. Tirole and Guesnerie (1981) construct such a process for
linear taxes based on gradient projections, leading to an ordinary
differential equation. The environment with nonlinear taxes is significantly
more challenging as now the whole tax function is evolved as opposed
to just one linear tax -- we therefore have the steepest descent
path in a space of functions. We have constructed such a process that
leads, within each step, to a well-behaved partial differential equation.
One advantage of this process is that it uses, similarly to the variational
approach, only the current information about the economy, such as
elasticities and the density of incomes. The second advantage is that
it gives rise to one of the most regular and well-behaved mathematical
objects -- the heat equation -- allowing for a wealth of characterizations.
Additionally, our analysis shows that continuous version of the iterative
fixed point method commonly used for computing optimal taxes (see,
e.g., Brewer, Saez, and Shepard 2010) can be represented within each
step as a heat equation.

Our analysis of the static taxation problem from the dynamical point
of view, the construction of the gradient flow, and the analysis of
the optimum and the gradient flow from the dynamical perspective are
new to the mechanism design and optimal taxation literature. 

More broadly, there are several papers that are connected to our work.
Bolton and Harris (2010) is perhaps the most comparable in terms of
the approach. They associate a dynamic risk sharing rule with that
of the static problem. Furthemore, Bolton and Harris (2010) obtain
an elegant asymptotic expansion of the dynamic problem around a myopic
optimum showing how the static problem is modified by the dynamic
correction terms. 

Sonnenschein (1981, 1982) and Artzner, Simon, and Sonnenschein (1986)
derive a heat equation as a gradient process of the firms adjusting
the commodity they produce by maximizing the rate of change in profit
subject to a quadratic cost of adjustment. McCann (2014) argues that
this result is a precursor to some of the results on the gradient
flows in the optimal transport literature.\footnote{See also Blanchet and Carlier (2015).}

Some of the techniques that we use have parallels in the optimal transport
literature (see, e.g., Villani (2003)) in which there is a renewal
of interest in economics (see, e.g., early work of Chiappori, McCann,
and Nesheim (2010), a comprehensive book by Galichon (2016), or a
review in the context of matching models by Chiappori and Salanie
(2016)). One important result in the optimal transport literature
is showing that a Fokker-Planck equation arises as a gradient flow
in a Wasserstein space (Jordan, Kinderlehrer, and Otto (1998)). In
mechanism design, optimal transport has been recently used in the
context of multi-dimensional screening by Figalli, Kim, and McCann
(2011) and Daskalakis, Deckelbaum, and Tzamos (2017).

Our analysis is interesting also from the purely mathematical point
of view. The heat equation is a classical object in mathematics and
physics (indeed, so classical that Fourier's specific way of solving
it via trigonometric series is the origin of harmonic analysis, an
entire subfield of mathematics). Its importance suggests that it should
indeed appear in a wide variety of settings. Here, we present a new
setting in which the heat equation naturally arises.

The paper is structured over two main parts that while connected are
self-contained. Section \ref{sec:The-optimal-tax} deals with the
optimal nonlinear tax. Section \ref{sec:Gradient-flow} deals with
the gradient flow of taxes.

\section{\label{sec:Environment}Environment}

For clarity of exposition, we start by presenting the simplest economic
environment of tax mechanisms. 

\subsection{Individuals}

Agents are characterized by an exogenous and fixed productivity type
$\theta\in\Theta\subset\mathbb{R}_{+}$. Preferences over consumption
$c$ and labor effort $l$ are represented by the utility function
$U\left(c,l\right)=c-v\left(l\right)$, where the disutility of labor
effort $v$ is twice continuously differentiable, increasing and strictly
convex. The government levies a tax liability $T:\mathbb{R}_{+}\rightarrow\mathbb{R}$
which can be an arbitrarily non-linear function of the individual's
labor income $y=\theta\times l$. The agent's budget constraint is
$c=\theta l-T(\theta l)$.

The optimization problem of an individual with type $\theta$ reads:
\begin{eqnarray}
 & \underset{l\ge0}{\max} & \theta l-T(\theta l)-v(l).\label{eq:IndividualProblem}
\end{eqnarray}
We denote by $l(\theta,T)\in\mathbb{R}_{+}$ the argmax of this problem
and by $y(\theta,T)\in\mathbb{R}_{+}$ the agent's labor income. For
ease of notation, when there is no ambiguity we remove the argument
$T$ from these variables and write them as $l\left(\theta\right)$,
$y\left(\theta\right)$.

\paragraph{Characterization of individual behavior.}

Assuming that the tax function $T$ is continuously differentiable,
labor income $y\left(\theta\right)\equiv y(\theta,T)$ is characterized
by the first-order condition: 
\begin{eqnarray}
1-T^{\prime}(y\left(\theta\right)) & =v^{\prime}\left(\frac{y\left(\theta\right)}{\theta}\right)\frac{1}{\theta} & .\label{eq:first order condition}
\end{eqnarray}
We assume that no individual $\theta$ is indifferent between two
or more incomes in the initial equilibrium: for all $\theta$, the
individual problem (\ref{eq:IndividualProblem}) has a unique global
maximum given the tax system $T$. It is straightforward then to show
that there is a one-to-one map between productivity types $\theta$
and pre-tax incomes $y(\theta)$.

\paragraph{Productivity and income distributions.}

We denote by $H\left(\theta\right)$ the c.d.f. of $\theta\in\Theta$,
and by $h\left(\theta\right)$ the corresponding density function.
We assume that the set $\Theta$ is a compact interval of $\mathbb{R}_{+}$,
and that the density of types $h$ is equal to zero at the boundaries
of $\Theta$. We also denote by $\Phi\left(y\right)$ and $\phi\left(y\right)$
the c.d.f. and the p.d.f. of incomes $y\in\mathbb{Y}\subset\mathbb{R}_{+}$.
We assume that the density of incomes $\phi$ is continuous and bounded
away from zero on any finite interval $[\underline{y},\bar{y}]\subset\mathbb{Y}$
with $\underline{y}>0$.\footnote{Our results can be straightforwardly generalized to the case of types
and incomes in the whole space $\mathbb{R}_{+}$ by using an increasing
sequence of compact sets $\Theta,\mathbb{Y}\subset\mathbb{R}_{+}$.}

\subsection{Government}

We define government revenue by 
\begin{eqnarray}
R\left(T\right) & = & \int_{\Theta}T(y(\theta,T))\text{d}H\left(\theta\right).\label{eq:SocialObjective}
\end{eqnarray}
It is equal to the sum over agents $\theta\in\Theta$ of the tax liability
on their income $y(\theta,T)$, taking into account their optimizing
behavior given the tax system $T$.\footnote{It is immediate to extend our results to the case of the government
maximizing a social welfare function.}

We denote by 
\[
p\left(y\right)=-\frac{\partial\ln(1-T^{\prime}\left(y\right))}{\partial\ln y}=\frac{yT^{\prime\prime}\left(y\right)}{1-T^{\prime}\left(y\right)}
\]
the local rate of progressivity of the tax schedule. It is equal to
(minus) the elasticity of the retention rate $1-T^{\prime}(y)$ with
respect to income $y$ (see, e.g., Musgrave and Thin (1948)).

\subsection{Taxable income elasticity}

The structural labor supply elasticity of agent $\theta$ with respect
to the retention rate $r\left(\theta\right)=1-T^{\prime}(y(\theta))$
is defined by 
\[
e\left(\theta\right)=\frac{\partial\ln y\left(\theta\right)}{\partial\ln r\left(\theta\right)}=\frac{v^{\prime}(l\left(\theta\right))}{l\left(\theta\right)v^{\prime\prime}(l\left(\theta\right))}.
\]
We define the compensated labor income response along the non-linear
budget constraint by 
\begin{eqnarray}
\varepsilon\left(\theta\right) & = & \frac{y\left(\theta\right)}{r\left(\theta\right)}\frac{e\left(\theta\right)}{1+p(y\left(\theta\right))e\left(\theta\right)}.\label{eq:LaborSupplyElasticity}
\end{eqnarray}
Assuming that there is a one-to-one map between types $\theta$ and
incomes $y\left(\theta\right)$, we can equivalently denote the labor
supply responses $\varepsilon\left(\theta\right)$ by $\varepsilon\left(y\right)$.
In the sequel we use both notations interchangeably. We assume that
$\varepsilon\left(y\right)$ is continuous and bounded away from zero
on any subinterval $[\underline{y},\bar{y}]\subsetneq\mathbb{Y}$.

The variable $\varepsilon\left(\theta\right)$ is equal to the partial
derivative of labor income $y\left(\theta\right)$ with respect to
the retention rate $r\left(\theta\right)$, normalized by the term
$1+p(y\left(\theta\right))e\left(\theta\right)$. Intuitively, this
term accounts for the fact that the agent's labor income adjustment
in response to an exogenous increase in his retention rate $r\left(\theta\right)$
leads to an endogenous shift in his tax rate by $p(y\left(\theta\right))$,
which in turn causes a further labor income adjustment $e\left(\theta\right)$.
Solving for the fixed point leads the total response (\ref{eq:LaborSupplyElasticity}).

\section{Tax reforms and optimum taxation}

In this section we define a notion of local tax reforms, and derive
their effects on individual behavior and government revenue. We then
derive a formula for the optimal tax. All of the results in this section
are standard in the literature.

\subsection{Variations of taxes}

We define a direction of reform of the tax function $T$ as a continuously
differentiable function $\hat{T}:\,\mathbb{R}_{+}\rightarrow\mathbb{R}$.
The perturbed tax function is then $T+\mu\hat{T}$, where $\mu>0$
is the size of the reform in the direction $\hat{T}$. In this section,
we derive the first-order change in government revenue $R(T)$ in
response to the tax reform $\mu\hat{T}$ as $\mu\rightarrow0$. That
is, we compute the Gateaux differential of the functional $T\mapsto R(T)$,
formally defined by: 
\begin{eqnarray*}
\delta R(T,\hat{T}) & \equiv & \lim_{\mu\rightarrow0}\frac{R(T+\mu\hat{T})-R(T)}{\mu}.
\end{eqnarray*}

We first describe the impact of a tax reform $\hat{T}$ of the tax
schedule $T$ on the government revenue. 
\begin{lem}
\label{lem:GateauxDiffIncome} The Gateaux differential of government
revenue in the direction $\hat{T}$ is given by: 
\begin{eqnarray}
\delta R(T,\hat{T}) & = & \int_{\mathbb{Y}}\hat{T}\left(y\right)\phi\left(y\right)dy-\int_{\mathbb{Y}}T^{\prime}\left(y\right)\varepsilon\left(y\right)\hat{T}^{\prime}\left(y\right)\phi\left(y\right)dy.\label{eq:GateauxDiff_SocialWelfare-1}
\end{eqnarray}
\end{lem}
\begin{proof}
The Gateaux differential of the government revenue functional in the
direction $\hat{T}$ is given by:

\[
\delta R(T,\hat{T})=\int_{\mathbb{Y}}\hat{T}\left(y\right)\phi\left(y\right)dy+\int_{\mathbb{Y}}T^{\prime}\left(y\right)\delta y\left(\theta\right)\phi\left(y\right)dy,
\]

where the differential of the agent's income $\delta y\left(\theta\right)$
is given by the Gateaux differential of the agent's first order condition
(\ref{eq:first order condition}): 
\[
-T''(y\left(\theta\right))\delta y\left(\theta\right)-\hat{T}^{\prime}(y\left(\theta\right))=v''\left(\frac{y\left(\theta\right)}{\theta}\right)\frac{1}{\theta^{2}}\delta y\left(\theta\right).
\]
Solving for $\delta y$ and using the definition of elasticity $\varepsilon\left(y\right)$
gives $\delta y\left(\theta\right)=-\hat{T}^{\prime}(y\left(\theta\right))\varepsilon\left(y\left(\theta\right)\right).$ 
\end{proof}
Equation (\ref{eq:GateauxDiff_SocialWelfare-1}) shows that the first-order
effect of the tax reform $\hat{T}$ on government revenue, $\delta R(T,\hat{T})$,
is given by the sum of two terms. The first integral on the right
hand side is the statutory impact of the reform, i.e., the amount
of revenue raised mechanically by changing the tax payment of each
agent with income $y$ by $\hat{T}\left(y\right)$. This term simply
sums these additional tax payments over the whole population, using
the density of incomes $\phi\left(\cdot\right)$, assuming that everyone's
behavior remains unchanged following the reform. The second integral
is the excess burden, or the deadweight loss, of the tax reform. Specifically,
consider the agents who earned income $y$ in the initial equilibrium,
i.e., before the tax reform. An increase in their marginal tax rate
by $\hat{T}^{\prime}\left(y\right)$ lowers the labor income of these
agents by $\varepsilon\left(y\right)$, by definition of their labor
income response along the non-linear initial tax schedule. This in
turn reduces government revenue by a fraction $T^{\prime}\left(y\right)$
of this income loss. Summing these effects over the population using
the density of incomes $\phi\left(\cdot\right)$ leads to expression
(\ref{eq:GateauxDiff_SocialWelfare-1}).

Note that this expression gives the effects of tax reforms $\hat{T}\left(y\right)$
in closed-form since all of the variables (the tax schedule $T\left(y\right)$,
the labor supply elasticities $\varepsilon\left(y\right)$, the density
of incomes $\phi\left(y\right)$) are observed or can be estimated
empirically in the economy with the given tax $T\left(y\right)$.

\subsection{Optimal Tax}

Lemma \ref{lem:GateauxDiffIncome} provides a formula for the revenue
effects of any tax reform $\hat{T}$ in the economy starting from
any, optimal or suboptimal, tax schedule $T$.

We now obtain a characterization of the optimal (i.e., government
revenue maximizing) tax schedule $T_{*}$ by imposing that no tax
reform has a non-zero first-order effect on government revenue, i.e.,
$\delta R(T_{*},\hat{T})=0$ for all $\hat{T}:\mathbb{R}_{+}\rightarrow\mathbb{R}$.
Let $\varepsilon_{*}\left(y\right)$, $\phi_{*}\left(y\right)$, and
$\Phi_{*}\left(y\right)$ denote the compensated labor income response
and the p.d.f. and c.d.f. of incomes given the tax schedule $T_{*}$.

Integrating by parts the Gateaux differential of the government revenue
functional (\ref{eq:GateauxDiff_SocialWelfare-1})

\[
\delta R(T,\hat{T})=\int_{\mathbb{Y}}\hat{T}\left(y\right)\phi\left(y\right)dy+\int_{\mathbb{Y}}\hat{T}\left(y\right)\frac{d}{dy}\left(T^{\prime}\left(y\right)\varepsilon\left(y\right)\phi\left(y\right)\right)dy,
\]

and setting it to zero for any $\hat{T}\left(y\right)$ yields: 
\begin{equation}
0=\phi_{*}\left(y\right)+\frac{d}{dy}\left(T_{*}^{\prime}\left(y\right)\varepsilon_{*}\left(y\right)\phi_{*}\left(y\right)\right).\label{eq:FOC government}
\end{equation}

Integrating with respect to $y$ yields the the following Corollary
(due to Mirrlees 1971, Diamond 1998, Saez 2001). 
\begin{cor}
\label{cor:The-optimal-tax}The optimal tax schedule $T_{*}$ satisfies:
for all $y\in\mathbb{Y}$, 
\begin{eqnarray}
T_{*}^{\prime}\left(y\right) & = & \frac{1}{\varepsilon_{*}\left(y\right)}\,\frac{1-\Phi_{*}\left(y\right)}{\phi_{*}\left(y\right)},\label{eq:OptimumTaxSchedule}
\end{eqnarray}
\end{cor}
Formula (\ref{eq:OptimumTaxSchedule}) shows that the optimal marginal
tax rate at income level $y$ is the product of two terms. First,
it is proportional to the inverse compensated labor income response
to tax rates at income $y$, $1/\varepsilon_{*}\left(y\right)$: the
higher the disincentive effect of marginal tax rates, the lower the
optimal tax rate. The second term is related to the hazard rate of
the income distribution, $(1-\Phi_{*}\left(y\right))/(\phi_{*}\left(y\right))$
and is a benefit-cost ratio that measures the fraction of agents whose
tax liability increases by a lump-sum amount in response to a marginal
tax rate increase at income $y$, relative to the fraction of agents
whose labor supply is distorted.

It is important to note that, while certainly important to gain economic
insights, the formula (\ref{eq:OptimumTaxSchedule}) is not a full
solution for the optimum tax as $y$ itself does depend on $T$ in
a nonlinear way since individuals optimize with respect to the tax
code. 

It is useful to discuss two different variants of the optimal tax
formula used in the literature. Saez (2001) derives the formula that
is identical to (\ref{eq:OptimumTaxSchedule}). On the left hand side
of this formula is the marginal tax on income which is the main object
of interest. The right hand side is, however, defined over endogenous
variables -- the elasticity along the nonlinear budget constraint
$\varepsilon_{*}$$\left(y\right)$, the density $\phi_{*}\left(y\right)$
and the c.d.f. $\Phi_{*}\left(y\right)$ of incomes are evaluated
at the optimum and hence themselves depend on the income tax schedule
$T_{*}\left(y\right)$. The Diamond (1998) formula, in the case of
iso-elastic preferences for labor $v\left(l\right)=\frac{l^{1+1/\epsilon}}{1+1/\epsilon}$,
is given by 
\[
\frac{T_{*}^{\prime}\left(y\left(\theta\right)\right)}{1-T_{*}^{\prime}\left(y\left(\theta\right)\right)}=\left(1+\frac{1}{\epsilon}\right)\frac{1-F\left(\theta\right)}{\theta f\left(\theta\right)}.
\]

The right hand side is a closed-form expression, since the distribution
of types $\theta$ (cdf $F$, pdf $f$) is exogenous and the elasticity
$\epsilon$ is given. However, the left hand side is not a closed
form expression for the income tax schedule $T\left(\cdot\right)$.
Indeed, note that the left hand side gives the marginal tax rate faced
by a type $\theta$. But this tax rate is evaluated at the income
$y\left(\theta\right)$ that the agent earns given the (optimal) tax
schedule. This variable is endogenous, and is given as the implicit
solution to the first order condition of the agent, which obviously
does not give $y\left(\theta\right)$ as a function of $\theta$ in
closed form (even with the isoelastic functional form of the disutility
of labor). In other words, the Diamond formula gives the tax on the
type $\theta$. We are interested in the tax schedule in the space
of incomes. However, the relationship between the income and the type
is unknown. In fact, it is the essence of the nonlinear income tax
problem that the types are unobservable and the tax schedule is over
incomes. The Diamond and the Saez formulas are essentially identical
to each other and neither provides a full characterization of the
optimum.

\section{The optimal tax and the fairness principle\label{sec:The-optimal-tax}}

In this section, we provide a new property of the optimum -- the
fairness principle -- viewing the static optimal taxation problem
from a dynamical point of view. 

First, we rewrite the optimal static tax formula in an operator notation.
Consider a second order differential operator $L=\frac{\partial}{\partial y}\left(\sigma_{*}\left(y\right)\frac{\partial}{\partial y}\right)$,
where $\sigma_{*}\left(y\right)=\varepsilon_{*}\left(y\right)\phi_{*}\left(y\right).$
The optimal tax in (\ref{eq:OptimumTaxSchedule}) is then given by
\[
LT_{*}=-\phi_{*}.
\]

It is known that operators of this type can be associated with the
heat kernel (Grigor'yan 2009). We also show in Section \ref{sec:Gradient-flow}
that the heat equation arises from studying the gradient flow of the
tax reform, that is, the trajectory of the steepest increase in revenue.
The static optimal tax is a stationary point of such a dynamical system. 

For the analysis of the optimal tax we proceed as follows. We first
introduce the heat kernel and show that the optimal tax satisfies
a certain invariance property -- the fairness principle -- with
respect to this object. Specifically, the optimal tax at a given income
can be represented as the weighted (by the heat kernel) average of
the optimal taxes at the other incomes. We then characterize the form
of the heat kernel and show that it behaves as a Gaussian average
that proportionally downweights more distant incomes. This weighting
is indexed by time (or, maybe more intuitively, by scale) where averaging
is done over a broader set of incomes. Importantly, this weighting
is tightly connected to each other at every scale, thus providing
one unifying weighting scheme at every scale (or set of incomes).
Finally, we show a representation of the optimal marginal income tax
as well as the higher derivatives of the optimal tax schedule in terms
of the levels of the optimal taxes.

We first introduce the heat kernel. Let $q_{t}\left(x,y\right)$ be
the heat kernel given by the solution to the Kolmogorov forward equation
\begin{equation}
\frac{\partial}{\partial t}q_{t}\left(x,y\right)=\frac{\partial}{\partial y}\left(\varepsilon_{*}\left(y\right)\phi_{*}\left(y\right)\frac{\partial}{\partial y}q_{t}\left(x,y\right)\right),\label{eq:Heat Kernel q_t}
\end{equation}

and $\lim_{t\rightarrow0}q_{t}\left(x,y\right)=\delta\left(x-y\right)$,
where $\delta$ is a Dirac delta function.\footnote{The kernel $q_{t}\left(x,y\right)$ solves both the forward and the
backward Kolmogorov equation $\partial_{t}q_{t}\left(x,y\right)=L_{y}q_{t}\left(x,y\right)=L_{x}q_{t}\left(x,y\right).$} 

An example of the heat kernel is a Gaussian in Figure 1, we show later
in this section that the heat kernel satisfying equation (\ref{eq:Heat Kernel q_t})
behaves similarly to the Gaussian. The heat kernel $q_{t}\left(x,y\right)$,
for a given income $x$, is a function of two variables -- income
$y$ and time $t$. One can think of time as a different scale over
incomes. At time (scale) zero, the kernel puts the weight one on income
$x$. The larger times (scales) average and encompass a wider set
of incomes.

We now state and prove the new property of the optimal tax -- the
fairness principle.
\begin{prop}
\label{Prop: Fairness-principle-for}(Fairness principle for the optimum).
The optimal tax $T_{*}\left(y\right)$ is invariant under the heat
kernel $q_{t}\left(x,y\right)$ given by (\ref{eq:Heat Kernel q_t}),
for any $x\in Y$ and any $t>0$: 
\begin{equation}
T_{*}\left(x\right)=\int_{0}^{t}\int q_{s}\left(x,y\right)\phi_{*}\left(y\right)dyds+\int q_{t}\left(x,y\right)T_{*}\left(y\right)dy.\label{eq:Prop 1 invariance}
\end{equation}
\end{prop}
\begin{proof}
Consider the derivative 
\[
\frac{\partial}{\partial t}\int q_{t}\left(x,y\right)T_{*}\left(y\right)dy=\int\frac{\partial}{\partial t}q_{t}\left(x,y\right)T_{*}\left(y\right)dy=
\]
by Kolmogorov forward equation
\[
=\int\frac{\partial}{\partial y}\left(\varepsilon_{*}\left(y\right)\phi_{*}\left(y\right)\frac{\partial}{\partial y}q_{t}\left(x,y\right)\right)T_{*}\left(y\right)dy=
\]
integrating twice by parts and rearranging
\[
=\int q_{t}(x,y)\frac{\partial}{\partial y}\left(\varepsilon_{*}\left(y\right)\phi_{*}\left(y\right)\frac{\partial}{\partial y}T_{*}\left(y\right)\right)dy=
\]
using equation (\ref{eq:FOC government})
\[
=-\int q_{t}\left(x,y\right)\phi_{*}\left(y\right)dy.
\]
Integrating the equation
\[
\frac{\partial}{\partial t}\int q_{t}\left(x,y\right)T_{*}\left(y\right)dy=-\int q_{t}\left(x,y\right)\phi_{*}\left(y\right)dy,
\]
we get
\[
\int q_{t}\left(x,y\right)T_{*}\left(y\right)dy=T_{*}\left(x\right)-\int_{0}^{t}\int q_{s}\left(x,y\right)\phi_{*}\left(y\right)dyds.
\]
\end{proof}
This proposition states that the optimal tax $T_{*}\left(x\right)$
is fair in the following sense -- the optimal tax at a given income
$x$ is equal to the weighted (by the heat kernel) average of taxes
at other incomes and income densities. The level of the optimal tax
at a given income $x$ is determined as the average of taxes at all
other incomes and at \textit{all} times (scales) $t$. We start by
focusing on the more interesting second integral, $\int q_{t}\left(x,y\right)T_{*}\left(y\right)dy$,
in (\ref{eq:Prop 1 invariance}). Since this term arises from the
behavioral effect of taxation it means that the planner smoothes the
distortions or the deadweight loss of taxes at every scale $t.$ The
averaging is done with the 1-parameter family of local averaging functions
$q_{t}\left(x,y\right)$, where each function corresponds to one particular
instance of the underlying fairness principle. In other words, the
optimal tax wants to ensure that an agent at a given income $x$ is
paying roughly the average of the amount of taxes paid by people working
just a little less or just a little more (``little less'' or ``little
more'' is determined by the heat kernel that downweights the more
distant incomes). This quantifies an underlying notion of fairness
that is in no way built into the system. The first integral represents
the mechanical effect of raising taxes and is a weighted average of
all densities of incomes $y$ at all times (scales) $t$.\footnote{One reason why this term is somewhat less central is that, for small
$t$, it is of the smaller order then the second integral term (see
Proposition \ref{Fairness-gradient} and footnote \ref{fn:Short-term optimal tax}).} Of course, our term fairness has nothing to do with the notions of
the social welfare function or the redistributional preferenes for
the government. In fact, the government here maximizes revenues. Yet,
the most efficient way to raise revenue is to do it fairly in the
sense of equating (with certain weights) the tax at any income to
the taxes of incomes at all scales. 

The key to our results is to view the static optimal tax from the
dynamic point of view by associating the heat kernel $q_{t}\left(x,y\right)$
with the operator $L$ determining the optimal tax. The optimal tax
$T_{*}\left(x\right)$ is of course time-independent. Yet, one can
think of it as being invariant under a dynamic system that starts
from this tax and applies the heat kernel $q_{t}\left(x,y\right)$
to it. At time $0,$ the optimal tax is just equal to itself -- thus
being trivially fair. As time goes on, the heat kernel shows that
the optimal tax is fair in the sense that it is equal to the average
over an increasingly wider distribution $q_{t}\left(x,y\right)$.
One can alternatively think of the variable $t$ not as time but as
a scale -- and the tax being fair at each scale $t$ encompassing
the weight of more and more incomes. We expand on the dynamic system
interpretation of the optimum in the next section where the optimum
arises as a stationary point of a gradient flow.

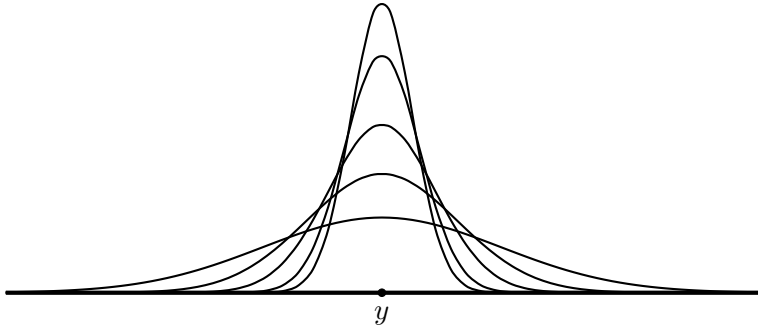
\begin{figure}[!h]
\begin{tikzpicture}
\draw [ultra thick] (0,0) -- (10,0);
\node at (-3,0) {};
\node at (5, -0.3) {$y$};
\filldraw (5,0) circle (0.05cm);
 \draw[scale=1,thick, samples=50,domain=0:10,smooth,variable=\x]  plot ({\x},{(5/sqrt(5/0.2))*exp(-0.2*(\x-5)*(\x-5))});
 \draw[scale=1,thick, samples=50,domain=0:10,smooth,variable=\x]  plot ({\x},{(5/sqrt(5/0.5))*exp(-0.5*(\x-5)*(\x-5))});
 \draw[scale=1,thick, samples=50,domain=0:10,smooth,variable=\x]  plot ({\x},{(5/sqrt(5))*exp(-(\x-5)*(\x-5))});
 \draw[scale=1,thick, samples=50,domain=0:10,smooth,variable=\x]  plot ({\x},{(5/sqrt(5/2))*exp(-2*(\x-5)*(\x-5))});
 \draw[scale=1,thick, samples=50,domain=0:10,smooth,variable=\x]  plot ({\x},{(5/sqrt(5/3))*exp(-3*(\x-5)*(\x-5))});
\end{tikzpicture} \caption{The heat kernel $q_{t}(x,y)$ for various times $t$ (smaller times
correspond to larger maxima).}
\end{figure}
The precise form of the fairness principle depends on the precise
form of $q_{t}$. We now turn to the more detailed characterization
of the heat kernel. 
\begin{prop}
(1) The heat kernel (\ref{eq:Heat Kernel q_t}) satisfies a Gaussian
upper bound, for any $t>0$: 
\[
q_{t}\left(x,y\right)\leq\frac{c_{1}}{\sqrt{t}}\exp\left(c_{2}t\right)\exp\left(-c_{3}\frac{(x-y)^{2}}{t}\right),
\]
for some positive constants $c_{1},$ $c_{2}$, $c_{3}$;

(2) For $t\rightarrow0,$
\[
q_{t}\left(x,y\right)\sim\frac{1}{\sqrt{4\pi\sigma_{*}\left(x\right)t}}\exp\left(-\frac{\left(y-x-\sigma_{*}'\left(x\right)t\right)^{2}}{4\sigma_{*}\left(x\right)t}\right),
\]
where $\sigma_{*}\left(x\right)=\varepsilon_{*}\left(x\right)\phi_{*}\left(x\right)$;

(3) For all point $x,y$ and all times $t,s>0$, the heat kernel satisfies
the semigroup property:
\[
q_{t+s}(x,y)=\int q_{t}(x,z)q_{s}(z,y)dz.
\]
\end{prop}
This proposition shows that the heat kernel essentially behaves similarly
to the properly scaled Gaussian average.\footnote{In one dimension that we have, the body of literature on parabolic
PDEs shows that for almost any modification of the problem, the heat
kernel looks and behaves exactly as a Gaussian -- this is true for
the heat kernels on arbitrary manifolds, for very wide classes of
conductivities $\sigma$, and for a very broad range of spaces (see
e.g, Grigor'yan (2009) or Bogachev, Krylov, Röckner, and Shaposhnikov
(2015) for extensive reviews).} The first part of the proposition shows that the heat kernel satisfies
a Gaussian upper bound for all times $t$. The second part of the
proposition shows that for short $t$, the heat kernel is exactly
the Gaussian. The third part of the proposition shows a fact of crucial
importance -- the heat kernels are tightly linked in at all time
scales. We now provide a more extensive discussion of these results.

The heat kernel $q_{t}$ is in general a very benign object that is
fairly easy to compute to any desired degree of accuracy. Three basic
properties are (1) $q_{t}(x,y)\geq0$, (2) $q_{t}(x,y)=q_{t}(y,x)$
and (3) preservation of integral mass 
\[
\int q_{t}(x,y)dy=1.
\]
We can thus, for a given point $x$, understand $q_{t}(x,y)$ as a
one-parameter family of probability distributions in the variable
$y$. This motivates understanding them as averaging objects. A classical
result of Aronson (1968) is that the heat kernel $q_{t}(x,y)$ on
a general manifold $M$ (satisfying very mild regularity assumptions)
satisfies what is called a Gaussian upper bound 
\[
q_{t}(x,y)\leq\frac{c_{1}}{t^{n/2}}\exp\left(-\frac{d(x,y)^{2}}{c_{2}t}\right),\quad\forall t>0,x,y\in M,
\]
where $d(x,y)$ is the geodesic distance between $x$ and $y$, and
$c_{1}$ and $c_{2}$ are positive constants. In particular, while
the heat kernel $q_{t}(x,y)$ may no longer look like a Gaussian centered
at $y$ having variance $t$, it certainly has the same decay behavior.
That is, it acts as a local averaging operator at scale $d(x,y)\sim\sqrt{t}$.
In other ways, the fairness principle averages the nearby income,
where the nearby is given by the scale $\sqrt{t}$.\footnote{If one is interested in the higher order expansions, those can be
straightforwardly derived in closed form to any order using the parametrix
method which represents the heat kernel as the sum of the Gaussian
and the higher order corrections (see, e.g., Friedman (2008)).} In the proposition, we use a slightly more general result in Metafune,
Ouhabaz, and Pallara (2011).

Finally, the classical results for the short-time asymptotics (see,
e.g., Varadhan (1967), Molchanov (1975), and Grigor'yan (2009)) imply
that for $t\rightarrow0,$ the heat kernel is the Gaussian with the
scale determined by the conductivity parameter $\varepsilon_{*}\left(x\right)\phi_{*}\left(x\right)$.
In the next section, we provide additional results and intuition for
this small time asymptotics.

Returning to the interpretation of $q_{t}$ as creating an averaging
operator at scale $\sim\sqrt{t}$, the third part of the proposition
shows the fairness principles for different scales are linked. A fact
of crucial importance is that they are tightly linked via what is
known as the semigroup property (see, e.g., Grigor'yan 2009). This
shows that the behavior of $q_{t}$ is tightly linked to both past
and future behavior of the heat kernel and is thus far from arbitrary.
That is, there is one unified weighting scheme at all income scales.

We now obtain from Proposition \ref{Prop: Fairness-principle-for}
the corresponding representation for the marginal tax. 
\begin{cor}
The fairness principle is invariant under differentiation and implies
that the marginal tax is given by 
\[
\frac{\partial}{\partial x}T_{*}\left(x\right)=\int_{0}^{t}\int\left(\frac{\partial}{\partial x}q_{s}\left(x,y\right)\right)\phi_{*}\left(y\right)dyds+\int\left(\frac{\partial}{\partial x}q_{t}\left(x,y\right)\right)T_{*}\left(y\right)dy,
\]
for any $x\in Y$ and any $t>0$
\end{cor}
The relevant quantity, $\partial_{x}q_{t}(x,y)$, is quite simple
to understand in one dimension: since $q_{t}(x,y)$ is a probability
distribution, its derivative has total integral 0. This means that
$\partial_{x}q_{t}(x,y)$ has a positive part and a negative part
with the same total $L^{1}-$mass and acts as a discrete differentiation
operator. We plot it in Figure 2 (together with the higher derivatives
of $q_{t}(x,y)$). In economic terms, the integral evaluates a weighted
average of taxes paid by individuals with slightly higher incomes,
subtracts a weighted average of taxes for individuals with slightly
lower incomes, and this results in the quantity determining the size
of the marginal tax $T'_{*}(x)$. Moreover, this fairness principle
for the marginal tax holds, as the original fairness principle, for
all $t>0$. 

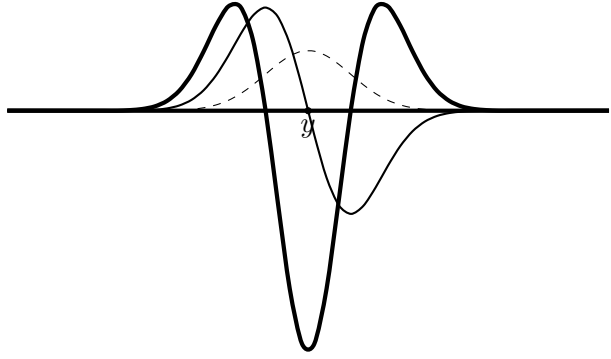
\begin{figure}[h!]
\begin{tikzpicture}[scale=0.8]
\draw [ultra thick] (0,0) -- (10,0);
\node at (-5,0) {};
\node at (5, -0.3) {$y$};
\filldraw (5,0) circle (0.05cm);
 \draw[scale=1,dashed, samples=50,domain=0:10,smooth,variable=\x]  plot ({\x},{exp(-(\x-5)*(\x-5))});
 \draw[scale=1,thick, samples=50,domain=0:10,smooth,variable=\x]  plot ({\x},{-(\x-5)*(4)*exp(-(\x-5)*(\x-5))});
 \draw[scale=1,ultra thick, samples=50,domain=0:10,smooth,variable=\x]  plot ({\x},{8*(\x-5)*(\x-5)*exp(-(\x-5)*(\x-5)) - 4*exp(-(\x-5)*(\x-5)) });
\end{tikzpicture} \caption{$q_{t}$ (dashed), its derivative $q_{t}'$ and its second derivative
(bold) $q_{t}''$.}
\end{figure}
It is an elementary mathematical fact that given any function $f\in C^{1}(\mathbb{R})$,
it is possible to change it ever so slightly into a function $f_{2}\in C^{1}(\mathbb{R})$
such that $f_{1}$ and $f_{2}$ give almost the same values everywhere
\[
\max_{x\in\mathbb{R}}{|f_{1}(x)-f_{2}(x)|}\leq\varepsilon
\]
but $f_{2}$ has a very different derivative 
\[
\max_{x\in\mathbb{R}}{|f_{1}'(x)-f_{2}'(x)|}\geq\frac{1}{\varepsilon}.
\]
Put differently, even a very good understanding of the optimal tax
code $T_{*}$ need not a priori translate into a good understanding
of the marginal tax $T_{*}'$. We show here that this is not the case,
the marginal tax is uniquely determined by the global behavior of
the optimal tax code.

One way of interpreting this statement is as follows: for any function
$f\in C^{1}$, $f$ at a point is the local average of its neighboring
values (indeed, this follows from continuity and boundedness and does
not require differentiability) 
\[
\lim_{t\rightarrow0^{+}}\int_{}{q_{t}(x,y)f(y)dy}=f(x).
\]
The fairness principle states that the optimal tax code satisfies
a much stronger relationship 
\[
T_{*}\left(x\right)=\int_{0}^{t}\int q_{s}\left(x,y\right)\phi_{*}\left(y\right)dyds+\int q_{t}\left(x,y\right)T_{*}\left(y\right)dy
\]
for \textit{all} times $t>0$. We can now see whether there is an
analogous result for the derivative. Differentiating the fairness
principle yields 
\[
\frac{\partial}{\partial x}T_{*}\left(x\right)=\int_{0}^{t}\int\left(\frac{\partial}{\partial x}q_{s}\left(x,y\right)\right)\phi_{*}\left(y\right)dyds+\int\left(\frac{\partial}{\partial x}q_{t}\left(x,y\right)\right)T_{*}\left(y\right)dy
\]
and it is of interest to understand whether there is an analogous
result for \textit{all} functions $f$. We perform the relevant computations,
for simplicity of exposition, for the heat kernel of the Laplacian
on $\mathbb{R}$, i.e. 
\[
q_{t}(x,y)=\frac{1}{\sqrt{4\pi t}}e^{-\frac{|x-y|^{2}}{4t}}.
\]
A simple computation shows that 
\[
\frac{\partial}{\partial x}q_{t}\left(x,y\right)=-\frac{(x-y)}{4\sqrt{\pi}t^{3/2}}e^{-\frac{|x-y|^{2}}{4t}}
\]
and from this we see that 
\[
\lim_{t\rightarrow0^{+}}\int\left(\frac{\partial}{\partial x}q_{t}\left(x,y\right)\right)f\left(y\right)dy=f'(x)
\]
This follows easily from a Taylor expansion of $f$ around $x$. While
this relationship is true for all continuous $f\in C^{1}(\mathbb{R})$,
it only holds in the limit $t\rightarrow0$. In contrast, as shown
by the derivative of the fairness principle, the optimal tax code
$T_{*}$ satisfies a relationship of this type for \textit{all} $t>0$.
The same principle holds for higher derivatives, and these results
can be obtained in the same manner.

Summarizing, in this section we derive a new characterization of the
optimal tax. The fairness principle, while still of course not a closed-form
solution, provides a new set of insights on the nature of the optimal
tax.

\section{\label{sec:Gradient-flow}Gradient flows of taxes}

We now turn to the analysis of taxes from a different point of view.
We construct a dynamic system, a gradient flow, which starts at any
(optimal or suboptimal) tax function and then changes the tax system
in the direction of the increased revenues. The optimal tax is a stationary
point of this system. 

In this section, we use the straightforward adaptation of notation
in Section \ref{sec:Environment} to index the relevant variables
by time. 

We start by formally defining the gradient flow. 
\begin{defn}
\label{def:Gradient flow}For all $t\geq0$ and $y\in\mathbb{Y}$,
the \emph{gradient flow} of the government revenue functional $R(T_{t})$
is defined as the dynamical system: 
\begin{eqnarray}
\frac{\partial T_{t}\left(y\right)}{\partial t} & = & \phi_{t}\left(y\right)+\frac{\partial}{\partial y}[T_{t}^{\prime}\left(y\right)\varepsilon_{t}\left(y\right)\phi_{t}\left(y\right)],\label{eq:Heat PDE}
\end{eqnarray}
where $\phi_{t}$ is governed by $T_{t}$ according to the change
of variables $\phi_{t}\left(y\left(\theta\right)\right)=(y_{t}^{\prime}\left(\theta\right))^{-1}h\left(\theta\right)$
and by equation (\ref{eq:first order condition}).\footnote{If the government is social welfare maximizing, it is immediate to
show that the gradient flow has the form of $\frac{\partial T_{t}\left(y\right)}{\partial t}=\phi_{t}\left(y\right)\left(1-\gamma_{t}\left(y\right)\right)+\frac{\partial}{\partial y}[T_{t}^{\prime}\left(y\right)\varepsilon_{t}\left(y\right)\phi_{t}\left(y\right)],$
where $\gamma_{t}\left(y\right)$ is the social marginal utility of
income (see Diamond (1975)).} 
\end{defn}

\subsection{\label{subsec:Mathematical foundations}Mathematical foundations}

\paragraph{Finite-dimensional spaces.}

Gradient flows are natural mathematical objects attached to functions
or functionals, mapping to real numbers. For simplicity, start with
a differentiable function $V:\mathbb{R}^{n}\rightarrow\mathbb{R}$
and define the gradient flow as a curve $x:\left[0,\infty\right]\rightarrow\mathbb{R}^{n}$
starting at some point $x_{0}\in\mathbb{R}^{n}$ with the property
that the curve always flows in the direction of steepest descent of
$V$. Intuitively, this direction is determined by the gradient of
$V$. Formally, we want to choose the vector $\hat{x}\in\mathbb{R}^{n}$
with $\Vert\hat{x}\Vert_{\ell^{2}}=1$ that minimizes 
\begin{eqnarray*}
\lim_{\mu\rightarrow0}\,\frac{1}{\mu}[\,V(x+\mu\hat{x})-V\left(x\right)\,] & = & \langle\,\nabla V\,,\,\hat{x}\,\rangle_{\ell^{2}},
\end{eqnarray*}
where the equality follows from the definition of the gradient. This
gives rise to an ordinary differential equation that describes the
law of motion of $x_{t}\in\mathbb{R}^{n}$ for $t\geq0$: 
\begin{eqnarray*}
\frac{d}{dt}x_{t} & = & -\,\nabla\,V(x_{t}),
\end{eqnarray*}
with the property that $V$ is decreasing along the flow of $x$ since
\[
\frac{d}{dt}V(x_{t})=\langle\nabla V,\frac{d}{dt}x_{t}\rangle=-\Vert\nabla V(x_{t})\Vert^{2}<0.
\]
While this model is rather classical and the existence and uniqueness
properties of the solution are well known, understanding the actual
dynamical behavior can pose considerable challenges (recent examples
being given by Tao (2017), Steinerberger (2018)).

\paragraph{Infinite-dimensional spaces.}

The very same principle can be applied in settings where the underlying
domain is not finite-dimensional but instead given by the space of
functions. We illustrate this with a representative example. We may
define a functional $\Psi$ by assigning to any twice-differentiable
\emph{function} $f\in\mathcal{C}^{2}\left(\mathbb{R},\mathbb{R}\right)$,
the number 
\begin{eqnarray*}
V\left(f\right) & = & \frac{1}{2}\int_{\mathbb{R}}\,|f^{\prime}\left(x\right)|^{2}\,\text{d}x.
\end{eqnarray*}
It is easy to show that the Gateaux differential of $V$ in the direction
$\hat{f}$ is given by $\delta V(f,\hat{f})=\int_{\mathbb{R}}f^{\prime}\left(x\right)\hat{f}^{\prime}\left(x\right)\text{d}x$.
An integration by parts implies that $\delta V(f,\hat{f})=-\int_{\mathbb{R}}f^{\prime\prime}\left(x\right)\hat{f}\left(x\right)\text{d}x$.
More generally, for any function $f\in\mathcal{C}^{2}\left(\mathbb{R}^{n},\mathbb{R}\right)$,
let $V\left(f\right)=\int_{\mathbb{R}^{n}}|\nabla f|^{2}$. We then
have $\delta V(f,\hat{f})=\int_{\mathbb{R}^{n}}\nabla f\cdot\nabla\hat{f}$.
By Green's first identity, this can be represented as a functional
$\hat{f}\mapsto\langle-\Delta f,\hat{f}\rangle_{L^{2}}$, where $\Delta$
denotes the Laplace operator, thus recovering the same structure as
above. That is, in order to flow in the direction of steepest descent
of the functional $V$, we must set 
\begin{eqnarray*}
\frac{\partial}{\partial t}f_{t} & = & \Delta\,f_{t}.
\end{eqnarray*}
This gives rise to a law of motion for the function $f$ characterized
by a parabolic PDE (namely, a heat equation). Needless to say, even
showing that all of these operations remain valid for any time $t>0$
is a difficult task, the theory of partial differential equations
being substantially more challenging than that of ordinary differential
equations.

\subsection{The gradient flow of taxes}

We now turn to the formal derivation of the the gradient flow (\ref{eq:Heat PDE}).
An integration by parts in the second integral of equation (\ref{eq:GateauxDiff_SocialWelfare-1})
implies that the impact of the tax reform $\hat{T}_{t}$ on government
revenue can equivalently be rewritten as 
\begin{eqnarray}
\delta R(T_{t},\hat{T}_{t}) & = & \int_{\mathbb{Y}}\Lambda_{t}\left(y\right)\hat{T}_{t}\left(y\right)dy,\text{ with }\Lambda_{t}\equiv\phi_{t}+\frac{\partial}{\partial y}[T_{t}^{\prime}\varepsilon_{t}\phi_{t}]\label{eq:Integration by parts Revenue}
\end{eqnarray}
Assuming that the space of functions $\mathcal{C}^{2}\left(\mathbb{R},\mathbb{R}\right)$
is endowed with the $L^{2}$ norm $\Vert T\Vert^{2}=\int\left(T\left(y\right)\right)^{2}dy$,
this can be expressed as $\langle\Lambda_{t},\hat{T}_{t}\rangle$.\footnote{We could have considered the weighted-$L^{2}$ norm $\Vert T\Vert^{2}=\int\kappa_{t}\left(y\right)\left(T\left(y\right)\right)^{2}\text{d}y$,
for some weights $\kappa_{t}\left(y\right)$ and this expression can
be represented as $\langle\kappa_{t}^{-1}\Lambda_{t},\hat{T}_{t}\rangle$
with the resulting gradient flow $\frac{\partial T_{t}\left(y\right)}{\partial t}=\left(\kappa_{t}\left(y\right)\right)^{-1}\phi\left(y\right)+\left(\kappa_{t}\left(y\right)\right)^{-1}\frac{\partial}{\partial y}[T_{t}^{\prime}\left(y\right)\varepsilon_{t}\left(y\right)\phi_{t}\left(y\right)]$.
The analysis for this case is identical. There is a re-interpretation
of such a weight as simply changing the metric of the underlying manifold
$\mathbb{R}$: put differently, one can interpret everything as an
equal-weight problem on a curved geometry; heat and associated processes
are not very sensitive to ``curving'' (heat propagates on a plane
and on a sphere in roughly the same sense). See, e.g., Taylor (1996),
Hörmander (2003), and Grigor'yan (2009). } Therefore, the gradient flow we obtain in this case can be written
as the dynamical system (\ref{eq:Heat PDE}).

The gradient flow (\ref{eq:Heat PDE}) can be equivalently derived
as the solution to the problem of choosing the trajectory of the tax
schedule $t\mapsto T_{t}$ that maximizes at each instant $t$ the
increase in government revenue: 
\begin{eqnarray*}
 & \underset{T_{t}}{\max} & \frac{\partial}{\partial t}\,R(T_{t}),
\end{eqnarray*}
in the $L^{2}$ norm. Specifically, the evolution of government revenue
$R(T_{t})$ over time for a given trajectory $T_{t}$ is given by
\begin{eqnarray*}
\frac{\partial}{\partial t}R(T_{t}) & = & \frac{\partial}{\partial t}\int\,T_{t}(y_{t}\left(\theta\right))\,\text{d}H\left(\theta\right)\ \,=\int\left[\frac{\partial T_{t}}{\partial t}(y_{t}\left(\theta\right))+T_{t}^{\prime}\left(y_{t}\left(\theta\right)\right)\frac{\partial y_{t}}{\partial t}\left(\theta\right)\right]dH\left(\theta\right).
\end{eqnarray*}
Imposing that the individual's first-order condition remains satisfied
over time requires (by differentiation of (\ref{eq:first order condition}))
that 
\[
\frac{\partial y_{t}}{\partial t}\left(\theta\right)=-\varepsilon_{t}\left(\theta\right)\frac{\partial T_{t}^{\prime}(y_{t}\left(\theta\right))}{\partial t}.
\]
That is, at each instant, individual $\theta$ adjusts his income
in the opposite direction and proportionally to the change in the
marginal tax rate that he faces.

Plugging this equation back into the law of motion of government revenue
and integrating the second term by parts leads to 
\begin{eqnarray*}
\frac{\partial}{\partial t}R(T_{t}) & = & \int\,\frac{\partial T_{t}}{\partial t}\,\{\,\phi_{t}(y)+\frac{\partial}{\partial y}[\,T_{t}^{\prime}(y)\,\varepsilon_{t}\left(y\right)\,\phi_{t}(y)\,]\,\}\,\text{d}y\ \,=\left\langle \Lambda_{t}\left(y\right),\frac{\partial T_{t}\left(y\right)}{\partial t}\right\rangle .
\end{eqnarray*}
This expression is maximized when $\frac{\partial T_{t}}{\partial t}(y)=\,\phi_{t}(y)+\frac{\partial}{\partial y}[T_{t}^{\prime}(y)\varepsilon_{t}\left(y\right)\phi_{t}(y)]$,
thus leading to (\ref{eq:Heat PDE}).

For another way to understand the economic meaning of (\ref{eq:Heat PDE}),
consider the problem of choosing the tax reform $\hat{T}$ that maximizes
the increase in government revenue, subject to the following quadratic
cost of reforming the tax payments:\footnote{Again, we can use weights $\kappa\left(y\right)>0$ in the cost function
(e.g., they can be equal to the density function $\phi_{t}\left(y\right)$)
and derive essentially the same results.} 
\begin{eqnarray*}
 & \underset{\hat{T}}{\max} & \delta R(T_{t},\hat{T})\,-\,\frac{1}{2}\int_{\mathbb{Y}}(\hat{T}\left(y\right))^{2}\,\text{d}y.
\end{eqnarray*}
Using the representation (\ref{eq:Integration by parts Revenue}),
the solution is given by 
\begin{eqnarray*}
\hat{T}_{t}\left(y\right) & = & \Lambda_{t}\left(y\right).
\end{eqnarray*}
Now, the law of motion of the tax schedule in the small interval of
time $\delta t\rightarrow0$ is given by $T_{t+\delta t}\left(y\right)=T_{t}\left(y\right)+\hat{T}_{t}\left(y\right)\delta t$,
or $\hat{T}_{t}\left(y\right)=\frac{\partial T_{t}(y)}{\partial t}$.
We therefore obtain the gradient flow (\ref{eq:Heat PDE}).

\subsection{\label{subsec:Short-term-evolution-and}Short-term evolution and
fairness property}

In this section, we describe the short-term evolution of the tax schedule
$T_{\tilde{t}}$ under the gradient flow by solving the heat equation
(\ref{eq:Heat PDE}) over a short time interval $\left[t,\tilde{t}\right]$.
We derive a version of the fairness principle that now applies to
any point on the path of evolution of the tax system.
\begin{prop}
\label{Fairness-gradient}(Fairness principle for the gradient flow).
Consider any initial time $t$ with the corresponding tax profile
$T_{t}\left(y\right)$, density of incomes $\phi_{t}\left(y\right)$,
elasticity $\varepsilon_{t}\left(y\right)$, and conductivity $\sigma_{t}\left(y\right)=\phi_{t}\left(y\right)$$\varepsilon_{t}\left(y\right)$.
Then, for small $\tilde{t}$, the tax $T_{\tilde{t}}\left(y\right)$,
generated by the gradient flow (\ref{eq:Heat PDE}), is given by a
weighted Gaussian average of the incomes: 
\[
T_{\tilde{t}}\left(x\right)\sim\left(\tilde{t}-t\right)\phi_{t}\left(x\right)+\int_{\underline{y}}^{\bar{y}}q_{t,\tilde{t}}\left(x,y\right)T_{t}(y)dy,
\]

where $q_{t,\tilde{t}}\left(x,y\right)=\frac{1}{\sqrt{4\pi\sigma_{t}\left(x\right)\left(\tilde{t}-t\right)}}\exp\left(-\frac{\left(y-x-\sigma_{t}'\left(x\right)\left(\tilde{t}-t\right)\right)^{2}}{4\sigma_{t}\left(x\right)\left(\tilde{t}-t\right)}\right)$.
\end{prop}
\begin{proof}
The proof uses the Feynman-Kac formula for path integrals (Lorinczi,
Hiroshima, and Betz 2011) to study the short-time behavior of solutions
of equations of the type (\ref{eq:Heat PDE}). Let $B\left(s\right)$
denote the diffusion process that satisfies the SDE\textbf{.} 
\[
dB_{s}=\sigma'_{s}\left(B_{s}\right)ds+\sqrt{2\sigma_{s}\left(B_{s}\right)}dW_{s},
\]
where $W$ is a Brownian motion. We then have 
\begin{eqnarray*}
T_{\tilde{t}}\left(x\right) & =\mathbb{E}\left[\int_{t}^{\tilde{t}}\phi_{s}\left(B\left(s\right)\right)ds\right]+\mathbb{E}[T_{t}(B_{\tilde{t}})],
\end{eqnarray*}
where the expectation runs over the diffusion process $B_{s}$, started
in $y$ and running up to $\tilde{t}$.\footnote{Technically, we need to specify boundary conditions for the Brownian
motion $B(s)$, however, since we are only using Brownian motion for
very small times $\tilde{t}$, it does not matter very much whether
we prescribe absorbing boundary conditions corresponding to Dirichlet
conditions or reflecting boundary conditions corresponding to Neumann
conditions. } We can now perform a Taylor expansion of these quantities. Up to
the first order, diffusivity is constant at a certain scale. The short-time
asymptotics for Brownian motion is then given by a Gaussian distribution
$B_{\tilde{t}}\sim x+\sigma_{t}'\left(x\right)\left(\tilde{t}-t\right)+\sqrt{2\sigma_{t}\left(x\right)}W_{t}\sim N\left(x+\sigma_{t}'\left(x\right)\left(\tilde{t}-t\right),2\sigma_{t}\left(x\right)\left(\tilde{t}-t\right)\right)$
and 
\begin{eqnarray*}
\text{distribution of }B_{\tilde{t}} & \sim & \frac{1}{\sqrt{4\pi\sigma_{t}\left(x\right)\left(\tilde{t}-t\right)}}\exp\left(-\frac{\left(y-x-\sigma_{t}'\left(x\right)\left(\tilde{t}-t\right)\right)^{2}}{4\sigma_{t}\left(x\right)\left(\tilde{t}-t\right)}\right).
\end{eqnarray*}
Then, up to a first order for $\tilde{t}$ small, 
\[
\int_{t}^{\tilde{t}}\phi_{t}\left(B\left(s\right)\right)ds\sim\left(\tilde{t}-t\right)\phi_{t}\left(x\right).
\]
This implies that 
\[
T_{\tilde{t}}\left(x\right)\sim\left(\tilde{t}-t\right)\phi_{t}\left(x\right)+\mathbb{E}[T_{t}(B_{\tilde{t}})]
\]
\begin{equation}
\sim\left(\tilde{t}-t\right)\phi_{t}\left(x\right)+\int_{\underline{y}}^{\bar{y}}\frac{1}{\sqrt{4\pi\sigma_{t}\left(x\right)\left(\tilde{t}-t\right)}}\exp\left(-\frac{\left(y-x-\sigma_{t}'\left(x\right)\left(\tilde{t}-t\right)\right)^{2}}{4\sigma_{t}\left(x\right)\left(\tilde{t}-t\right)}\right)T_{t}(y)dy.\label{eq:Tax over short time}
\end{equation}
\end{proof}
This equation extends the notion of fairness that we derived for the
optimal tax to that of the trajectory of the gradient flow of taxes.
The government that considers a tax reform in the direction of maximizing
revenues changes the tax such that the new, evolved tax $T_{\tilde{t}}$$\left(x\right)$
is equal to the weighted average of the initial taxes $T_{t}\left(y\right)$
plus another source term $\phi_{t}\left(x\right)$ that evaluates
the mechanical effect of the revenues collected. In other words, the
gradient flow wants to ensure that an agent at a given income $x$
is paying roughly the average of the amount of taxes paid by people
working just a little less or just a little more. The conductivity
$\varepsilon_{t}\left(x\right)\phi_{t}\left(x\right)$ determines
the scale of the Gaussian and the drift correction and the function
$\phi_{t}\left(x\right)$ determines the asymptotic deviation.\footnote{\label{fn:Short-term optimal tax}Of course, this equation also holds
for the optimal tax, as it is a stationary point of the gradient flow
in which both the taxes and agents' behavior no longer change: $T_{*}\left(x\right)=t\phi_{*}\left(x\right)+\int_{\underline{y}}^{\bar{y}}\frac{1}{\sqrt{4\pi\sigma_{_{*}}\left(x\right)t}}\exp\left(-\frac{\left(y-x-\sigma_{*}'\left(x\right)t\right)^{2}}{4\sigma_{*}\left(x\right)t}\right)T_{*}(y)dy+o\left(t\right)$. } 

Note, that the property (\ref{eq:Tax over short time}) holds for
any starting time $t$ on the gradient flow trajectory and that we
average over the known and given parameters $\varepsilon_{t}\left(x\right)$
and $\phi_{t}\left(x\right)$ evaluated at the time $t$. That is,
it is a closed-form expression. This characterization is valid for
the short time $\tilde{t}$ as these parameters are essentially frozen
over that short time interval. For the large time, it is also a fairly
good approximation of the heat kernel, but with the conductivity and
the source needing to be adjusted as the agents change their behavior
in response to the evolution of the tax function. In contrast, the
results for the optimal tax are derived for any time (scale) $t$
as well as for the short term asymptotics but include the elasticities
and the densities evaluated at the optimum. 

There is also a sense in which the gradient flow acts on the points
that deviate most from this principle. It is encoded in the equation
\[
T_{\tilde{t}}\left(x\right)\sim\left(\tilde{t}-t\right)\phi_{t}\left(x\right)+\mathbb{E}[T_{t}(B_{\tilde{t}})].
\]
Say the source is $\phi_{t}\equiv0$ and suppose that $T_{t}=1$ everywhere
except in $x=0$. Then $T_{\tilde{t}}\left(0\right)=1$ as the averaging
corrects this deviation. We elaborate the discussion of the smoothing
properties of the gradient flow in the next section.

\subsection{Gradient flow as a heat equation \label{subsec:Operator-Splitting}}

In this section, we propose a construction of a trajectory of the
tax reform and show how the gradient flow of the revenue functional
gives rise to the heat equation.

Equation (\ref{eq:Heat PDE}) changes the tax schedule in favor of
increasing government revenue, letting $\phi_{t}$ and $\varepsilon_{t}$
be endogenously driven by $T_{t}$ -- that is, taking into account
the fact that the density and the elasticity change in response to
the evolution of taxes.

We propose to evolve the system separately (this underlying idea is
a straightforward application of ``operator splitting''). The simplest
instance of this idea is as follows. Suppose we are given a system
of ordinary differential equations given as 
\[
\frac{d}{dt}u(t)=(A+B)u(t),
\]
then the solution is given by the matrix exponential $u(t)=e^{t(A+B)}u(0)$.
A formal Taylor series expansion suggests that 
\begin{align*}
e^{t(A+B)} & =\mbox{Id}+t(A+B)+\mathcal{O}(t^{2})\\
 & =(\mbox{Id}+tA)(\mbox{Id}+tB)+\mathcal{O}(t^{2})\\
 & =e^{tA}e^{tB}u(0)+\mathcal{O}(t^{2}).
\end{align*}
These computations are purely formal but they do suggest that, at
least for small values of $t$, we may solve the system by first evolving
along the simpler system $\dot{u}(t)=Au$ and then along the system
$\dot{u}=Bu$ and alternate in this manner (Varga (1962), Glowinski
and Osher (2016)).

We apply the very same method in our problem: more precisely, we fix
the distribution of incomes $\phi_{t}$ and the elasticity $\varepsilon_{t}$
for a short period of time $\delta t$, evolve $T_{t}$, and then
re-compute $\phi_{t+\delta t}$ and $\varepsilon_{t+\delta t}$ based
on the new tax function $T_{t+\delta t}$.\footnote{This also can be regarded as a classical numerical technique for systems
of this type.} In standard situations, this procedure will converge to a solution
path of the dynamical system as $\delta t\rightarrow0$ (Glowinski,
Osher, and Yin (2016)). The operator splitting technique also has
a natural economic meaning. The government evolves taxes in the direction
of increased revenues, keeping the density of agents' incomes and
elasticities fixed at their value observed in the current economy.
That is, the government evaluates the changes in revenues under the
current information given by the exogenous sufficient statistics evaluated
at a given initial time.

This implies equation (\ref{eq:Heat PDE}) is a heat equation (with
source term $\phi$ and local conductivity $\sigma=\varepsilon\phi$),
i.e. a PDE of the form 
\begin{equation}
\frac{\partial T}{\partial t}=\phi\left(y\right)+\frac{\partial}{\partial y}[\sigma\left(y\right)\frac{\partial T}{\partial y}]\label{eq:heat equation for splitting frozen}
\end{equation}
and guarantees in particular that the problem always has a solution
(the heat equation being well-posed). We also note that our assumption
that the density tends to 0 at the boundary of the interval implies
that no boundary conditions need be imposed. Since heat equations
are among the most well-known and well-behaved partial differential
equations, we can apply standard mathematical results to obtain theoretical
properties of the evolution of the tax schedule over time.

Fixing $\phi$ and $\varepsilon$, letting $T$ evolve for a short
amount of time, then unfreezing $\phi$ and $\varepsilon$ and recomputing
it can be regarded as a classical example of operator splitting. While
the analysis of convergence of this dynamic system is outside the
scope of the paper, one can expect that for sufficiently short time
steps, the solution converges to the global optimum at a great level
of generality. For example, the review of Glowinski, Osher, and Yin
(2016, p.13) concludes: ``Last but not least, operator splitting
algorithms are theoretically attractive because they converge under
very few assumptions.'' More broadly, the splitting procedure we
use is similar in spirit to the ones used in physical sciences where
the split terms correspond to different physical processes -- for
example, splitting convection from diffusion (see, e.g. MacNamara
and Strang (2016)) or splitting fast from slow variables.

\subsection{Smoothing properties of the heat equation}

We now present various smoothing properties of the gradient flow arising
within a step of the operator splitting; this gradient flow is realized
by a heat equation. Most of the results in this section are straightforward
adaptations of very classical results for the heat equation. Since
the heat equation can alternatively be realized as a gradient flow
in a certain Sobolev-type space, it has a very nice underlying smoothing
structure which is reflected in a large number of beneficial mathematical
properties. Moreover, since the heat equation is one of the most studied
(and well-behaved) objects in mathematics, the list of the useful
properties is very large. We therefore focus only on some of them
in this section that show how the heat equation smoothes the underlying
functions. Various other properties could be of interest but are outside
the scope of this paper. For the rest of the section to ease the notational
burden we suppress the indexing of taxes with $t$ whenever it does
not cause confusion. 

Since the initial tax system can be arbitrary, we need to introduce
the proper reference point for the analysis of smoothing of a given
tax system. We proceed as follows. Since $\phi(y)$ and $\sigma\left(y\right)$
are fixed, the solution will converge to a fixed point as $t\rightarrow\infty$
and the unique fixed point is given by setting $\partial T/\partial t=0$
resulting in 
\begin{equation}
0=\phi\left(y\right)+\frac{\partial}{\partial y}[\sigma\left(y\right)\frac{\partial T\left(y\right)}{\partial y}].\label{eq:definition of tau}
\end{equation}
We define the solution to this equation as $\tau\left(y\right)$.
This stationary point is the solution to the problem of maximizing
tax revenue, conditional on keeping the density of agents' incomes
and elasticities fixed. In other words, this is the tax schedule under
exogenous sufficient statistics that would be optimal if these sufficient-statistic
variables were fixed and equal to their value observed in the current
economy, i.e., given the current tax $T_{t}$. 

It is important to note, however, that we are not interested in letting
the gradient flow for all $t\rightarrow\infty$ and thus finding $\tau$
but only evolving it for a very small time $t$ (so as to respect
the operator splitting). In this sense, $\tau$ serves as a proper
reference point to describe various smoothing properties of the gradient
flow. 

We first relate the gradient flow arising from revenue maximization
to another gradient flow with strong smoothing properties.
\begin{prop}
Equation (\ref{eq:definition of tau}) coincides with the gradient
flow of the functional $\mathcal{J}$ given by the weighted Sobolev-type
seminorm $H^{1}$ 
\begin{eqnarray*}
\mathcal{J}\left(T\right) & = & \frac{1}{2}\int_{}^ {}{\sigma(y)\left(T'\left(y\right)-\tau'\left(y\right)\right)^{2}dy}
\end{eqnarray*}
\end{prop}
\begin{proof}
Let us compute the directional derivative in the direction of a function
$w$ evaluated at $T$ 
\[
\delta\mathcal{{J}}(w)=\lim_{\varepsilon\rightarrow0}{\frac{\mathcal{{J}}(T+\varepsilon w)-\mathcal{{J}}(T)}{\varepsilon}}.
\]
We see that 
\begin{align*}
\lim_{\varepsilon\rightarrow0}{\frac{\mathcal{{J}}(T+\varepsilon w)-\mathcal{{J}}(T)}{\varepsilon}} & =\int\sigma(y)(T'\left(y\right)-\tau'\left(y\right))w'\left(y\right)dy\\
 & =\int w\left(y\right)\frac{\partial}{\partial y}\left(\sigma(y)\left(\tau'\left(y\right)-T'\left(y\right)\right)\right)dy.
\end{align*}
This shows that the negative gradient flow is given by 
\[
\frac{{\partial T}}{\partial t}=-\frac{\partial}{\partial y}\left(\sigma(y)\left(\tau'\left(y\right)-T'\left(y\right)\right)\right)=-\frac{\partial}{\partial y}\left(\sigma(y)\tau'\left(y\right)\right)+\frac{\partial}{\partial y}\left(\sigma(y)T'\left(y\right)\right).
\]
Using the equation for $\tau$, we see that $-\frac{\partial}{\partial y}\left(\sigma(y)\tau'\left(y\right)\right)=\phi(y)$
and we have established the desired claim. 
\end{proof}
This proposition has an interesting economic meaning. This gradient
flow has the effect of trying to smooth out rough irregularities in
the difference between $T$ and $\tau$ -- a large value of $|T'(y)-\tau'\left(y\right)|$
implies the existence of a large value of $\left(T'(y)-\tau'\left(y\right)\right)^{2}$
and the flow is trying to decrease this as quickly as possible. We
note that the quantity $\sigma(y)\geq0$ serves as a natural weighting
measure: if $\sigma(y)$ is large, then irregularities in that region
count even more severely and are dampened quicker than in regions
where $\sigma(y)$ is very small.

Moreover, if $T(y)$ has large amounts of strong oscillations or maybe
even discontinuous jumps, then the gradient flow acts strongest on
those parts first. This leads to the following proposition which follows
from the classical result on parabolic equations. 
\begin{prop}
Let $T\left(t,y\right)$ denote the solution of (\ref{eq:heat equation for splitting frozen}).
If $\sigma\left(y\right)$ is bounded away from $0$ and $T\left(0,y\right)$
is bounded, then $T\left(t,y\right)$ is infinitely differentiable
for any $t>0$. 
\end{prop}
This implies that the tax reform viewed as the gradient flow leads
to the continuous tax systems. Moreover, the gradient flow has the
effect of mollifying any tax scheme instantaneously.

The next proposition shows that the gradient flow smoothes a measure
of variability of the tax schedule, the squared deviation from the
limiting stationary solution. Moreover, such smoothing is exponential.
As we discussed above, this result is not about the convergence to
the optimal tax $T_{*}$ but rather about the smoothing behavior of
the gradient flow at each step of the operator splitting. As the initial
tax function $T\left(0,y\right)$ and hence the associated agents'
behavior (that determine $\phi\left(y\right)$ and $\varepsilon\left(y\right)$)
can be arbitrary, the correct reference point for this smoothing behavior
is the corresponding stationary solution $\tau\left(y\right)$.
\begin{prop}
Let $T\left(t,y\right)$ be the solution to (\ref{eq:heat equation for splitting frozen}),
$T\left(0,y\right)$ be an arbitrary initial tax schedule, and $\tau\left(y\right)$
be the solution to the stationary problem (\ref{eq:definition of tau}),
and $\lambda_{1}$ be the first eigenvalue of the associated Sturm-Liouville
operator: 
\begin{equation}
\frac{\partial T}{\partial t}=\phi\left(y\right)+\frac{\partial}{\partial y}[\sigma\left(y\right)\frac{\partial T}{\partial y}].\label{eq: Sturm Liouville PDE}
\end{equation}
Then, $\forall t>0:$ 
\[
\int_{a}^{b}{(T(t,y)-\tau\left(y\right))^{2}dy}\leq e^{-2\lambda_{1}t}\int_{a}^{b}{(T(0,y)-\tau(y))^{2}dy}.
\]
\end{prop}
\begin{proof}
We use the standard Sturm-Liouville theory (see Zettl (2010), Teschl
(2012), Titchmarsh (1962)) to prove this result. This equation can
be studied by first solving for the stationary problem 
\[
0=\phi\left(y\right)+\frac{\partial}{\partial y}[\sigma\left(y\right)\frac{\partial\tau}{\partial y}].
\]
Subtracting both equations leads to an equation for $z(t,x)=T(t,x)-\tau(x)$
given by 
\[
\frac{\partial z}{\partial t}=\frac{\partial}{\partial y}[\sigma\left(y\right)\frac{\partial z}{\partial y}].
\]
It remains to study problems of this type. We will do so by studying
the spectrum of the associated differential operator $H$ given by
\[
Hz=-\frac{\partial}{\partial y}[\sigma\left(y\right)\frac{\partial z}{\partial y}]
\]
or, in other words, we study the problem $Hz=\lambda z$. This eigenvalue
problem has a discrete sequence of admissible values $\lambda$ for
which the equation has a solution: these values $0<\lambda_{1}<\lambda_{2}<\dots$
are the eigenvalues of this operator of Sturm-Liouville type, the
corresponding solutions will be denoted by $\eta_{1},\eta_{2},\dots$
and are assumed to be $L^{2}-$normalized, i.e. $\|\eta_{n}\|_{L^{2}}=1$.
We note that $\lambda_{0}=0$ is a special value and $\eta_{0}=\mbox{const}$.
We see that these eigenfunctions are necessarily orthogonal in $L^{2}(a,b)$
since, again by integration by parts, 
\begin{align*}
\int_{a}^{b}{\eta_{k}(y)\eta_{\ell}(y)dy}= & \frac{1}{\lambda_{k}}\int_{a}^{b}{-(\sigma(y)\eta_{k}'(y))'\eta_{\ell}(y)dy}\\
= & -\frac{1}{\lambda_{k}}\int_{a}^{b}{-\sigma(y)\eta_{k}'(y)\eta_{\ell}'(y)dy}\\
= & \frac{1}{\lambda_{k}}\int_{a}^{b}{(-\sigma(y)\eta_{\ell}'(y))'\eta_{k}(y)dy}\\
= & \frac{\lambda_{\ell}}{\lambda_{k}}\int_{a}^{b}{\eta_{k}(y)\eta_{\ell}(y)dy}
\end{align*}
If $k\neq\ell$, then the factor in front of the integral is different
from 1 and the integral is therefore 0. This together with the completeness
of the system of eigenfunctions in $L^{2}$ allows us to expand an
arbitrary initial function $T(0,x)$ into a series 
\[
T(0,x)=\tau\left(x\right)+\sum_{k=1}^{\infty}{\left\langle z(0,x),\eta_{k}(x)\right\rangle \eta_{k}(x)}.
\]
We will abbreviate $a_{k}=\left\langle z(0,x),\phi_{k}(x)\right\rangle $
for simplicity of exposition. We then claim that 
\[
T(t,x)=\tau\left(x\right)+\sum_{k=1}^{\infty}{a_{k}e^{-\lambda_{k}t}\eta_{k}(x)}
\]
is a solution of the partial differential equation (\ref{eq: Sturm Liouville PDE}).
This can be verified by computing 
\begin{align*}
\left(\frac{d}{dt}-\frac{d}{dx}\sigma(x)\frac{d}{dx}\right)\sum_{k=1}^{\infty}{a_{k}e^{-\lambda_{k}t}\eta_{k}(x)} & =\sum_{k=1}^{\infty}{a_{k}\left(\frac{d}{dt}-\frac{d}{dx}\sigma(x)\frac{d}{dx}\right)e^{-\lambda_{k}t}\eta_{k}(x)}.
\end{align*}
The separation of variables implies that 
\[
\left(\frac{d}{dt}-\frac{d}{dx}\sigma(x)\frac{d}{dx}\right)e^{-\lambda_{k}t}\eta_{k}(x)=e^{-\lambda_{k}t}\bigg(-\lambda_{k}\eta_{k}(x)-\frac{d}{dx}\sigma(x)\frac{d}{dx}\eta_{k}(x)\bigg)=0
\]
as desired. Since we now have a complete description of a solution,
we can analyze the convergence to the limiting function arising for
$t\rightarrow\infty$ at a greater level of detail: we have 
\begin{align*}
\int_{a}^{b}{(T(t,y)-\tau\left(y\right))^{2}dy} & =\int_{a}^{b}{\left(\sum_{k=1}^{\infty}{a_{k}e^{-\lambda_{k}t}\eta_{k}(y)}\right)^{2}dy}=\int_{a}^{b}{\sum_{k,\ell=1}^{\infty}{a_{k}e^{-\lambda_{k}t}\eta_{k}(y)a_{\ell}e^{\lambda_{\ell}t}\eta_{\ell}(y)}dy}\\
 & =\int_{a}^{b}{\sum_{\ell=1}^{\infty}{a_{\ell}^{2}e^{-2\lambda_{l}t}\eta_{l}(y)^{2}}dy}=\sum_{\ell=1}^{\infty}{a_{\ell}^{2}e^{-2\lambda_{k}t}}\leq e^{-2\lambda_{1}t}\sum_{\ell=1}^{\infty}{a_{\ell}^{2}}.
\end{align*}
We note that 
\[
\sum_{\ell=1}^{\infty}{a_{\ell}^{2}}=\int_{a}^{b}{(T(0,y)-\tau(y))^{2}dy}
\]
and that we have therefore shown that 
\[
\int_{a}^{b}{(T(t,y)-\tau\left(y\right))^{2}dy}\leq e^{-2\lambda_{1}t}\int_{a}^{b}{(T(0,y)-\tau(y))^{2}dy}.
\]
Appealing to the classical Rayleigh-Ritz formula, we see that 
\begin{align}
\lambda_{1} & =\inf_{\int_{a}^{b}f(y)dy=0}{\frac{\int_{a}^{b}{\sigma(y)f'(y)^{2}dy}}{\int_{a}^{b}{f(y)^{2}dy}}}\label{eq:lambda1}
\end{align}
where the last step follows from the classical Neumann eigenvalue
computation for the homogeneous rod (see Courant and Hilbert (1989)).
This shows that for sufficiently regular values of $\phi(y)$, we
can expect $\lambda_{1}>0$ and therefore the distance to $\tau$$\left(y\right)$
undergoes exponential decay.
\end{proof}
Finally, note that the short term asymptotics results in Section \ref{subsec:Short-term-evolution-and}
do not require the use of operator splitting methods.

\section{Conclusion}

We have shown that dynamical methods can provide new insights on the
analysis of static optimal taxation problems. We show that the heat
kernel and the heat equation are intimately connected with the analysis
of this classic problem. Since the heat equation is one of the most
basic mathematical objects, it possesses a variety of useful properties
that can enrich our understanding of the mechanism design problems.
One such new characteristic is the fairness property where the taxation
system implies that a tax on a given income is a proper average of
taxes at other incomes and income densities. The fairness principle
does not stem from any desire of the planner to be fair but rather
shows that the most efficient way to raise revenues is to equalize
a properly weighted average of taxes at any given income. We have
shown that this principle holds for both the optimum at any scale
and for the gradient flow for the short time approximation. The derivation
of the fairness principle fundamentally relies on the dynamical view
of the static problem.

We now briefly outline some extensions. We considered the simplest
taxation environment to starkly highlight the main contributions of
the paper. Several extensions, some of which we already discussed
in the body of the paper are immediate. Consider a model where utility
is not quasi-linear, the government maximizes social welfare function,
agents have multidimensional types (not incomes). This model delivers
a gradient flow that is very similar to the one that we constructed,
with the properly modified elasticities and the social welfare weights.
More broadly, one can consider a variety of other mechanism design
problems -- nonlinear pricing, matching, etc., -- where gradient
flows, or more broadly viewing a static problem from the point of
view of a dynamical system, may be useful in characterizing the static
optimum and the evolution of the locally improving suboptimal policies. 

\newpage{}

\section*{References}

\[
\]

Aronson, Don. Non-negative solutions of linear parabolic equations,
Ann. Sci. Norm. Sup. 22, p. 607--694. 1968.

Artzner, Philippe, Carl P. Simon, and Hugo Sonnenschein. \textquotedbl{}
Convergence of Myopic Firms to Long-Run Equilibrium via the Method
of Characteristics.\textquotedbl{} Models of Economic Dynamics. Springer,
Berlin, Heidelberg, 157-183. 1986.

Bierbrauer, Felix, and Pierre Boyer. \textquotedbl Politically feasible
reforms of non-linear tax systems.\textquotedbl{} 2018.

Blanchet, Adrien, and Guillaume Carlier. \textquotedbl Optimal transport
and Cournot-Nash equilibria.\textquotedbl{} Mathematics of Operations
Research 41, no. 1: 125-145. 2015.

Bogachev, Vladimir I., Nicolai V. Krylov, Michael Röckner, and Stanislav
V. Shaposhnikov. Fokker-Planck-Kolmogorov Equations. Vol. 207. American
Mathematical Soc., 2015.

Bolton, Patrick, and Christopher Harris. ``The dynamics of optimal
risk sharing.'' No. w16094. National Bureau of Economic Research.
2010.

Brewer, Mike, Emmanuel Saez, and Andrew Shephard. \textquotedbl Means-testing
and tax rates on earnings.\textquotedbl{} Dimensions of Tax Design:
the Mirrlees Review. 2010.

Chiappori, Pierre-André, Robert J. McCann, and Lars P. Nesheim. \textquotedbl Hedonic
price equilibria, stable matching, and optimal transport: equivalence,
topology, and uniqueness.\textquotedbl{} Economic Theory 42, no. 2:
317-354. 2010.

Chiappori, Pierre-André, and Bernard Salanié. \textquotedbl The econometrics
of matching models.\textquotedbl{} Journal of Economic Literature
54, no. 3 : 832-61, 2016.

Courant, Richard and David Hilbert, Methods of mathematical physics.
Vol. I. Interscience Publishers, Inc., New York, N.Y., 1953.

Daskalakis, Constantinos, Alan Deckelbaum, and Christos Tzamos. \textquotedbl Strong
Duality for a Multiple-Good Monopolist.\textquotedbl{} Econometrica
85, no. 3: 735-767. 2017.

Diamond, Peter A. \textquotedbl{} A many-person Ramsey tax rule.\textquotedbl{}
Journal of Public Economics 4.4: 335-342. 1975.

Diamond, Peter A. \textquotedbl{} Optimal income taxation: an example
with a U-shaped pattern of optimal marginal tax rates.\textquotedbl{}
American Economic Review: 83-95. 1998.

Figalli, Alessio, Young-Heon Kim, and Robert J. McCann. \textquotedbl When
is multidimensional screening a convex program?.\textquotedbl{} Journal
of Economic Theory 146, no. 2: 454-478. 2010.

Friedman, Avner. Partial differential equations of parabolic type.
Courier Dover Publications, 2008.

Galichon, Alfred. Optimal transport methods in economics. Princeton
University Press. 2016.

Glowinski, Roland, Stanley J. Osher, and Wotao Yin, eds. Splitting
Methods in Communication, Imaging, Science, and Engineering. Springer,
2017.

Grigor'yan, Alexander. Heat Kernel and Analysis on Manifolds, AMS/IP
Studies in Advanced Mathematics. 2009.

Golosov, Mikhail, Aleh Tsyvinski, and Nicolas Werquin. A variational
approach to the analysis of tax systems. No. w20780. National Bureau
of Economic Research, 2014.

Hellwig, Martin F. \textquotedbl Incentive problems with unidimensional
hidden characteristics: A unified approach.\textquotedbl{} Econometrica
78, no. 4: 1201-1237. 2010.

Hörmander, Lars. The Analysis of Linear Partial Differential Operators
I -- IV, Springer, 2003.

Jordan, Richard, David Kinderlehrer, and Felix Otto. \textquotedbl{}
The variational formulation of the Fokker-Planck equation.\textquotedbl{}
SIAM Journal on Mathematical Analysis 29.1: 1-17. 1998.

Jacquet, Laurence, and Etienne Lehmann. \textquotedbl Optimal Income
Taxation when Skills and Behavioral Elasticites are Heterogeneous.\textquotedbl{}
2015.

Kleven, Henrik Jacobsen. \textquotedbl Sufficient Statistics Revisited.\textquotedbl{}
2018.

Kleven, Henrik Jacobsen, and Claus Thustrup Kreiner. \textquotedbl The
marginal cost of public funds: Hours of work versus labor force participation.\textquotedbl{}
Journal of Public Economics 90, no. 10-11: 1955-1973. 2006.

Kleven, Henrik Jacobsen, Claus Thustrup Kreiner, and Emmanuel Saez.
\textquotedbl The optimal income taxation of couples.\textquotedbl{}
Econometrica 77, no. 2: 537-560. 2009.

Lörinczi, Jozsef, Fumio Hiroshima, and Volker Betz. Feynman-Kac-type
theorems and Gibbs measures on path space: with applications to rigorous
quantum field theory. Vol. 34. Walter de Gruyter, 2011.

MacNamara, Shev, and Gilbert Strang. In Glowinski, Roland, Stanley
J. Osher, and Wotao Yin. Operator splitting. In Splitting Methods
in Communication, Imaging, Science, and Engineering (pp. 95-114).
Springer, Cham. 2016.

Metafune, Giorgio, El Maati Ouhabaz, and Diego Pallara. \textquotedbl Long
time behavior of heat kernels of operators with unbounded drift terms.\textquotedbl{}
Journal of Mathematical Analysis and Applications 377, no. 1: 170-179.
2011.

McCann, Robert J. \textquotedbl{} Academic wages, singularities, phase
transitions and pyramid schemes.\textquotedbl{} Proceedings of the
International Congress of Mathematicians (Seoul 2014). Vol. 3. 2014.

Mirrlees, James A. \textquotedbl{} An exploration in the theory of
optimum income taxation.\textquotedbl{} The Review of Economic Studies
38.2: 175-208. 1971.

Molchanov, Stanislav A. \textquotedbl{} Diffusion processes and Riemannian
geometry.\textquotedbl{} Russian Mathematical Surveys 30, no. 1: 1-63.
1975.

Musgrave, Richard A., and Tun Thin. \textquotedbl{} Income tax progression,
1929-48.\textquotedbl{} Journal of Political Economy 56.6: 498-514.
1948.

Sachs, Dominik, Aleh Tsyvinski, and Nicolas Werquin. Nonlinear tax
incidence and optimal taxation in general equilibrium. No. w22646.
National Bureau of Economic Research. 2016.

Saez, Emmanuel. Using elasticities to derive optimal income tax rates.
The Review of Economic Studies, vol. 68, no 1, p. 205-229. 2001.

Saez, Emmanuel, and Stefanie Stantcheva. \textquotedbl Generalized
social marginal welfare weights for optimal tax theory.\textquotedbl{}
American Economic Review 106, no. 1: 24-45. 2016.

Saez, Emmanuel, and Stefanie Stantcheva. \textquotedbl A simpler
theory of optimal capital taxation.\textquotedbl{} Journal of Public
Economics 162: 120-142. 2018.

Salanie, Bernard. The economics of taxation. MIT press, 2011.

Scheuer, Florian, and Iván Werning. Mirrlees meets diamond-mirrlees.
No. w22076. National Bureau of Economic Research. 2016.

Sonnenschein, Hugo. \textquotedbl{} Price dynamics and the disappearance
of short-run profits: An example.\textquotedbl{} Journal of Mathematical
Economics 8.2: 201-204. 1981.

Sonnenschein, Hugo. \textquotedbl{} Price dynamics based on the adjustment
of firms.\textquotedbl{} The American Economic Review 72.5: 1088-1096.
1982.

Steinerberger, Stefan. \textquotedbl{} Fast escape in incompressible
vector fields.\textquotedbl{} Monatshefte für Mathematik 186, no.
3: 525-537. 2018.

Tao, Terence. \textquotedbl{} On the universality of potential well
dynamics.\textquotedbl{} arXiv preprint arXiv:1707.02389. 2017.

Taylor, Michael. Partial Differential Equations I-III, Springer. 1996.

Teschl, Gerald. Ordinary differential equations and dynamical systems.
Vol. 140. American Mathematical Soc., 2012.

Tirole, Jean, and Roger Guesnerie.\textquotedbl{} Tax reform from
the gradient projection viewpoint.\textquotedbl{} Journal of Public
Economics 15.3: 275-293. 1981.

Titchmarsh, Edward Charles. Eigenfunction expansions associated with
second-order differential equations. Part I. Second Edition Clarendon
Press, Oxford. 1962.

Varadhan, Sathamangalam R. Srinivasa. \textquotedbl{} On the behavior
of the fundamental solution of the heat equation with variable coefficients.\textquotedbl{}
Communications on Pure and Applied Mathematics 20, no. 2: 431-455.
1967.

Varga, Richard S. Matrix Iterative Analysis, New Jersey: Prentice-Hall,
1962

Villani, Cedric. Topics in optimal transportation. No. 58. American
Mathematical Soc., 2003.

Zettl, Anton. Sturm--Liouville Theory. Providence: American Mathematical
Society, 2005.
\end{document}